%% file: index.tex
\documentclass[10pt,conference,letterpaper]{IEEEtran}
\usepackage[utf8]{inputenc}
\usepackage[T1]{fontenc}
\usepackage{fixltx2e}
\usepackage{graphicx}
\usepackage{longtable}
\usepackage{float}
\usepackage{wrapfig}
\usepackage{soul}
\usepackage{textcomp}
\usepackage{marvosym}
\usepackage{latexsym}
\usepackage{amssymb}
\usepackage{url}
\usepackage[capitalize]{cleveref}
\usepackage{todonotes}
\usepackage{amsmath}
\usepackage{amssymb}
\usepackage[titlenotnumbered,tworuled,vlined]{algorithm2e}
\usepackage{xspace} 

\usepackage{tikz}
\usetikzlibrary{shapes,arrows,shadows,automata}

\usepackage{mathtools}
\DeclarePairedDelimiter\ceil{\lceil}{\rceil}

\usepackage{amsthm}

\newtheorem{lemma}{Lemma}
\newtheorem{corollary}{Corollary}
\newtheorem{theorem}{Theorem}
\newcommand{\PeerCensus}{PeerCensus\xspace}
\newcommand{\event}{\mathbf}%

\DeclareMathOperator{\rank}{rank}
\newcommand{\digest}[1]{\mathcal{D}(#1)}

\newcommand{\Wonline}[1]{\ensuremath{P(#1)}}

\newcommand{\voting}[1]{V_{#1}}
\newcommand{\onlinevoting}[1]{I_{#1}}
\newcommand{\hide}[1]{}

\newenvironment{algo}{%
  \algorithm
}{%
  \endalgorithm
}
\SetAlgoVlined
\DontPrintSemicolon
\SetCommentSty{emph}
\newcommand{\Protocol}[2]{\textbf{Protocol:} #1
\smallskip

#2}
\newcommand{\Specification}[2]{\textbf{Specification:} #1
\smallskip

#2}
\SetKwBlock{Init}{Initialization:}{}
\SetKwProg{OnEvent}{On Event}{:}{}
\SetKwComment{Comment}{$\triangleright$ }{}
\SetKwProg{OnCommit}{On Commit}{:}{}
\SetKwProg{Validate}{Validate}{:}{}
\SetKwBlock{State}{Shared State:}

\title{Bitcoin Meets Strong Consistency}
\author{\IEEEauthorblockN{Christian Decker}\IEEEauthorblockA{ETH Zurich\\cdecker@ethz.ch} \and  \IEEEauthorblockN{Jochen Seidel}\IEEEauthorblockA{ETH Zurich\\seidelj@ethz.ch} \and  \IEEEauthorblockN{Roger Wattenhofer}\IEEEauthorblockA{ETH Zurich\\wattenhofer@ethz.ch}}
\date{\today}

\begin{document}

\maketitle

\begin{abstract}
  The Bitcoin system only provides eventual consistency.
  For everyday life, the time to confirm a Bitcoin transaction is prohibitively slow.
  In this paper we propose a new system, built on the Bitcoin blockchain, which enables strong consistency.
  Our system, \PeerCensus, acts as a certification authority, manages peer identities in a peer-to-peer network, and ultimately enhances Bitcoin and similar systems with strong consistency.
  Our extensive analysis shows that \PeerCensus is in a secure state with high probability.
  We also show how Discoin, a Bitcoin variant that decouples block creation and transaction confirmation, can be built on top of \PeerCensus, enabling real-time payments.
  Unlike Bitcoin, once transactions in Discoin are committed, they stay committed.
\end{abstract}

\section{Introduction}
\label{sec-1}
Since its inception in 2008, the Bitcoin~\cite{nakamoto2008bitcoin} cryptocurrency has been steadily growing in popularity.
Today, Bitcoin has a market capitalization of about 5 billion USD. The Bitcoin network processes transactions worth approximately 60 million USD each day.

So, how usable are Bitcoins in everyday life?
While one certainly can buy a coffee with Bitcoins, a Bitcoin transaction is shockingly insecure when compared to a cash (or credit card) transaction.
Cash is exchanged on the spot with the coffee, and credit card companies are liable for fraud attempts.
Bitcoins are different, as the Bitcoin system only guarantees ``eventual consistency''.
The barista will serve a coffee in exchange for a signed Bitcoin transaction by the customer. 
However, a signed Bitcoin transaction is no guarantee that the Bitcoin transfer really takes place.

In order to get a better understanding, let us follow the path of our Bitcoin transaction.
First, the barista will inject the signed transaction into the Bitcoin network, which is a random-topology peer-to-peer network.
The correctness of the signature will be immediately verified by the peers that get the transaction.
Next, the transaction will be flooded within the Bitcoin network, such that all peers in the Bitcoin network have seen the transaction.
Eventually, the transaction will be included in a block, and finally the block will end up in the~blockchain.

While the problem of fraudulent customers also exists with cash or credit cards, Bitcoins allow fraud on a whole different level. The main issue are so-called double-spend attacks~\cite{bamert2013snack, karame2012two}. Our coffee consumer may simply spend the same money multiple times. In addition to signing the transaction for our barista, the customer may concurrently sign another transactions spending the same Bitcoins but with the customer himself as beneficiary. While the barista is injecting her transaction into the Bitcoin network, the customer is injecting his transaction into the Bitcoin network as well, quickly and with as many peers as possible. Both the original and the double transactions will spread in the Bitcoin network, but the double-spend was injected at multiple vantage points, so it will spread more quickly. A professional fraudulent customer will manage that the double-spend transaction is orders of magnitude more present in the Bitcoin network than the original transaction. As such the double transaction will be much more likely to end up in a block, and ultimately in the blockchain.

The problem is that the barista cannot verify the whole process in real time. While injecting a transaction into the Bitcoin network, and the verification of the signature by the first peer is a matter of seconds, all the other steps in the process take time. Flooding transactions in a network already is an operation which may take minutes, and a block is only generated every 10 minutes~\cite{nakamoto2008bitcoin}. However, with the current backlog,\footnote{\url{https://blockchain.info/unconfirmed-transactions}} it is unlikely that a transaction will be in the next block. Rather, a few blocks might be generated before our transaction (the original or the double) managed to be selected in a block, so for a low-value transaction like the payment of a coffee we can expect a delay of about 30 minutes.
In addition there is the problem of so-called blockchain forks~\cite{decker2013information}, i.e., two conflicting blocks may generated at roughly the same time, and only subsequent blocks will determine which of the blocks is part of the blockchain and which one is discarded.
Each subsequent block takes another 10 minutes, so in order to know that a transaction is confirmed, we may need to wait for several hours.
The Bitcoin system is a prime example of eventual consistency: Eventually Bitcoin has a consistent view of the transactions, but one can never be sure, and it may always happen that a blockchain fork will destroy a substantial amount of transactions, sometimes even multiple hours later~\cite{bip0050}.

Because of this we argue that the current version of Bitcoin is fundamentally flawed when it comes to real time transactions, where goods or services are instantly exchanged for Bitcoins.
How long should our barista wait until she is sure that the transaction will eventually be in the blockchain?
Waiting for more confirmations does reduce the probability of the transaction being reverted, but how safe is safe enough?
When should the seller release the goods or service to the buyer? 
Most vendors are probably unaware of this tradeoff between safety and time.
In order to use Bitcoin for real time exchanges, we need to completely abandon the weak concept of eventual consistency and instead embrace strong consistency.

In this work we propose \emph{\PeerCensus}, a system upon which strongly consistent applications can be built.
The basic idea is that Bitcoin's blockchain can be used to introduce and manage identities that participate in the system.

More precisely, \PeerCensus{} uses the blockchain as a way to limit and certify new identities joining the system.
This yields strong guarantees on the assignment of these identities to entities participating in it.
We stress that \PeerCensus is application agnostic, i.e., it does not manage any application specific information.
A single \PeerCensus{} instance may be shared by an arbitrary number of applications.
In particular \PeerCensus{} can be used to introduce strong consistency in Bitcoin.
For easier readability, we call the strongly consistent Bitcoin that uses \PeerCensus{} \emph{Discoin}.

Discoin does not rely on its own blockchain.
Instead, it can rely on a byzantine agreement protocol~\cite{castro1999practical,kotla2007zyzzyva,lamport1982byzantine} to commit transactions to the transaction history, effectively decoupling block generation from transaction confirmation and thus enabling safe and fast transactions.
Once a transaction is committed it cannot be reverted at any future time, a property we refer to as \emph{forward security}.
This is in contrast to Bitcoin, where confirmations are slow and can be reverted by a sufficiently strong attacker.


Our approach is also significant in light of the recent proliferation of alternative digital currencies, the so-called altcoins, all reliant on their own blockchain.
The creation of altcoins has had the effect of splitting resources among many blockchains, resulting in many smaller and consequently more easily attackable blockchains.
\PeerCensus, with its shared instance, allows the computational resources to be concentrated to a single blockchain, strengthening it against attacks.

Moreover, \PeerCensus enables experimental versions of Bitcoin to test protocol changes at a smaller scale before merging them with the main network.
This is an alternative to the approach of \cite{back2014sidechains}, which instead suggests to allow transactions between otherwise separate blockchains.

The security guarantees of \PeerCensus are extensively analyzed in \Cref{sec:analysis}, where we show that with high probability the system does not fail.
Furthermore, we outline how the current Bitcoin system can be migrated to Discoin running on top of \PeerCensus, gaining strong consistency and real-time payments as a result.
Migrating resources and blocks from Bitcoin allows us to maintain the momentum and the public acceptance Bitcoin has gathered over the years.
Our proposed migration method results in an instance of \PeerCensus that in expectation fails fewer than once every 7~million years.

\section{Overview}
\label{sec:overview}

Our main objective is to enable the creation of a cryptocurrency that provides forward security and supports fast confirmations.
We accomplish this goal by leveraging techniques from Bitcoin as well as byzantine agreement protocols, resulting in strong consistency guarantees.
%
Known agreement protocols are not applicable to a peer-to-peer environment in which Bitcoin operates, for three reasons: Openness, Sybil Attacks, and Churn.

\begin{itemize}
\item
  \emph{Openness}:
  The set of peers eligible to participate in the protocol changes over time, but previous protocols rely on a fixed set of participants.
\item
  \emph{Sybil attacks}:
  Entities may participate in the protocol with an arbitrary number of identities, effectively disrupting voting based agreement protocols.
\item
  \emph{Churn}:
  Peers may join or leave the system at arbitrary times, therefore the quorum size required for agreement cannot be constant.
\end{itemize}


\begin{figure}
 \centering
 \includegraphics[width=0.45\textwidth]{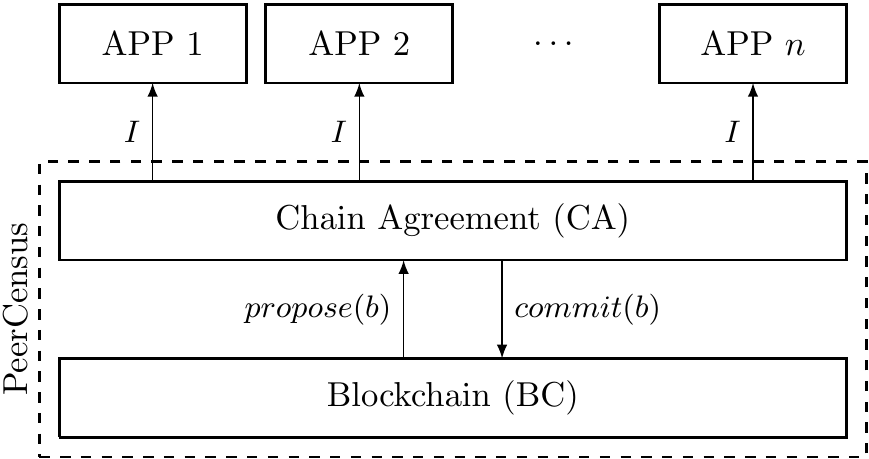}
\caption{The layout of the components and information flows.}
\end{figure}

Typical voting based agreement protocols, like PBFT~\cite{castro1999practical} and Zyzzyva~\cite{kotla2007zyzzyva}, require knowledge of the membership:
Before proceeding, the protocol must determine whether a sufficient number of participants voted.
This requirement is in stark contrast to the openness of a peer-to-peer setting.
Moreover, allowing unrestricted entry of new peers to the system creates the potential of Sybil attacks.
In a Sybil attack, a single entity poses as an arbitrary number of peers (by generating fake identities) and joins the system as distinct participants in order to subvert the system.
While the issue of churn has been addressed by previous agreement protocols (e.g., Secure Group Membership Protocol~\cite{reiter1996securegroup}), to the best of our knowledge Sybil attacks are left unaddressed by traditional agreement protocols.

Bitcoin introduced a novel use of Proof-of-Work systems, namely a \emph{blockchain} data structure, as a mechanism to deal with the problems caused by openness.
But in Bitcoin states can temporarily diverge, since each peer applies incoming operations to its local state without reaching any kind of agreement beforehand.
As a result, Bitcoin only guarantees eventual consistency, a property that is questionable for a protocol that is supposed to handle financial transactions.

In \PeerCensus we combine those approaches to obtain the best of both worlds:
Resilience to Sybil attacks and strong consistency.
Correspondingly, \PeerCensus consists of two core components:
the \emph{Blockchain (BC)} and the \emph{Chain Agreement~(CA)}.

The Blockchain's purpose is to mitigate Sybil attacks.
This is achieved by regulating the rate at which identities gain privileges within the system, and by ensuring that those privileges are not obtained by a single entity.
Peers are either \emph{non-voting} or \emph{voting} peers.
In particular, new peers start as non-voting until promoted to \emph{voting} by appending a block to the collaboratively maintained blockchain.
The rate at which blocks can be found in the network can be regulated so that new identities are promoted at a fixed rate, currently every 10 minutes.
Furthermore, the share of identities an entity may control converges to the share of computational resources it controls in the network.


The Chain Agreement on the other hand augments the system with strong consistency.
By virtue of the voting rights issued from the Blockchain, a byzantine agreement protocol can be used.
The CA's task is twofold.
One task is to track the system membership, i.e., which identities are currently online and participating.
This ensures that a voting based agreement protocol such as PBFT can function correctly.

The other task is to resolve conflicts in case of a blockchain fork, i.e., if multiple blocks are proposed for extending the blockchain, then only one of them will be committed.
Using standard agreement protocol techniques we immediately obtain strong consistency.
\PeerCensus guarantees that with high probability, an entity can only subvert the agreement if it controls a sufficiently large share of all resources.

Applications built on top of \PeerCensus may rely on the guarantees about the identity distribution and the membership.
To demonstrate how simple it is to build strongly consistent applications on top of \PeerCensus, we introduce a new cryptocurrency called Discoin.
Because of the \PeerCensus foundation, Discoin itself can rely on classical byzantine agreement protocols to atomically confirm transactions.
Transactions are proposed to the primary in Discoin, which assigns sequence numbers to them and attempts to commit them to the transaction history.
Since transactions are totally ordered, double-spends can be resolved locally, and upon committing all peers agree on a common transaction history.

Compared to the current Bitcoin system, Discoin and the underlying \PeerCensus system have several advantages:
\begin{itemize}
 \item A small blockchain since blocks only contain a single identity.
 \item Blockchain forks are resolved immediately when they occur.
 \item Confirmations are decoupled from blocks, enabling real-time confirmations.
 \item Since \PeerCensus tracks the participating identities, Discoin can distribute rewards and transaction fees to all participants instead of just the block finder.
\end{itemize}

Ultimately, \PeerCensus not only enables the creation of strongly consistent, but also simpler applications, by abstracting the dynamic membership.

\section{System Model}
\label{sec:model}

The setting in which \PeerCensus operates consists of the following three components:
a) a \emph{peer-to-peer} system,
b) the notion of \emph{controlling entities}, and
c) the notion of computational \emph{resources} at an entity's disposal.
The role of the peer-to-peer system is to execute the \PeerCensus protocol, whereas a controlling entity models an individual, possibly having control over several peers.
A \emph{Proof-of-Work} (PoW) mechanism (see \cref{sec:proof-of-work}) controls the entry rate of peers to the system to mitigate Sybil attacks.
In particular, the amount of PoWs a controlling entity $e$ can generate, and thus the number of peers controlled by $e$ entering the system, is dictated by the amount of (computational) \emph{resources} at $e$'s disposal.

\paragraph{Peers and Identities}

We denote by $P$ the set of peers that may join the network.
The \emph{identities} (IDs) of peers are established using public-key cryptography as follows:
When a peer $p \in P$ joins the network for the first time, $p$ generates a public-/private-keypair.
The identity of peer $p$ is its public key (or a derivative thereof).
We assume that there is no collision among the IDs chosen by the peers---in practice this is ensured by the assumption that obtaining the private key from the public key is computationally infeasible.
We do not require that IDs are ordered, and the outcome of \PeerCensus does not depend on the IDs chosen by the peers.

The system evolves in discrete unit time steps. 
At any given time, a peer $p \in P$ may either be online or offline, and we refer to the set of online peers at time $t$ by $\Wonline{t} \subseteq P$.
Offline peers may \emph{join} the network at arbitrary times, whereas online peers may \emph{leave} the network by either halting (voluntarily) or crashing (involuntarily) at any time.

Peers communicate via message passing in a point-to-point network.
This could either be viewed as having a completely connected communication graph, or by relaying messages among participants.
We simply assume that between any two online peers there is a channel which eventually delivers all messages.
The authenticity of every message is ensured by signing it with the sender's private key.


\paragraph{Controlling Entities}

The notion of collusion and control sharing among multiple peers is formalized by introducing \emph{controlling entities}.
Each peer $p$ is assigned to exactly one entity $e$ which controls its behavior.
The goal of $e$ is to steer $p$, hoping to maximize the entity's utility, i.e., entities are \emph{selfish}.

\paragraph{Resources}

In order to model computational limitations of entities, we introduce the notion of a \emph{computational unit-resource}, or \emph{resource} for short.
The set of unit-resources that will ever participate in the system is denoted by $R$, and $R(t) \subseteq R$ is the set of active resources at time $t$.
Every resource in $R$ is associated to exactly one entity which owns it.
All unit-resources are thought to possess the same computational power, and the more resources are active for an entity, the more computational tasks can be solved by that~entity.

Similarly to peers, resources may exit the system voluntarily or because of failure.
We assume that the failure and recovery probabilities of unit-resources are independent from their assignment to an entity.

\section{Dynamic Membership Protocol}
\label{sec-3}
\label{sec:protocol}


In this section we present the \PeerCensus{} protocol which provides a trustless decentralized certification authority for identities.
The \PeerCensus protocol consists of three layers, namely
\begin{itemize}
\item the Blockchain (BC) layer,
\item the Chain Agreement (CA) layer, and
\item the Application (APP) layer.
\end{itemize}

We now turn to describing each layer separately, starting with the Blockchain, which is based on a Proof-of-Work mechanism.

\hide{
  The fundamental layers are the BC and the CA layer, and the role of those two is to provide a resilient consensus abstraction to the AC layer.
  The BC consists of a bare bones version of the blockchain found in Bitcoin.
  Is sole purpose is to manage suffrage in the system.
  The CA's task is to track the blockchain status, break ties in case of a blockchain fork, and provide the interface for the AC layer.
  Using the primitives implemented by the BC and CA layers, applications like Bitcoin can be implemented on the AC layer.
  Those applications now receive the benefit of strong consistency properties.
}

\subsection{Blockchain (BC)}
\label{sec-3-2}
\label{sec:blockchain}



\paragraph*{Proof-of-Work Mechanisms}
\label{sec:proof-of-work}

An integral tool used in the Blockchain protocol is a so called \emph{Proof-of-Work} (PoW) mechanism.
This concept was introduced by Dwork and Naor in~\cite{dwork1992pricing}---we only give a brief overview in this subsection.
The key insight behind PoW mechanisms is that that the resources needed to solve computational puzzles are not easily acquired and may not be scaled at will.

A function $\mathcal F(d,c,x) \to \{\text{true, false}\}$, where $d$ is a positive number, and $c$ and $x$ are bit-strings, is called a \emph{PoW function} if it has following properties:

\begin{enumerate}
\item \label{item:pow-check} $\mathcal F(d,c,x)$ is fast to compute if $d, c$, and $x$ are given, and
\item \label{item:pow-work} for fixed parameters $d$ and $c$, finding $x$ so that $\mathcal F(d,c,x) = \text{true}$ using a unit-resource is distributed with $\exp(1/d)$, i.e., computationally difficult but feasible.
\end{enumerate}

We refer to the parameters $d, c$, and $x$ as \emph{difficulty}, \emph{challenge}, and \emph{nonce}, respectively.
For example, $\mathcal F$ might return true if and only if the output of some cryptographic hash function to the concatenation $x|c$ starts with at least $d$ zeroes.

The PoW mechanism issues a difficulty and a challenge pair $(d,c)$.
A nonce $x$ for which $\mathcal F(d,c,x) = \text{true}$ is called a \emph{Proof-of-Work} for $(d,c)$.
In our model, computational resources are required to find such an $x$.
We assume that no entity has an unfair advantage in finding a PoW.
Furthermore, we expect the PoW mechanism to automatically adjust the difficulty%
\footnote{The PoW mechanism used by Bitcoin accomplishes this (cf.~\cite{nakamoto2008bitcoin}).}
between consecutive $(d,c)$ pairs so that the expected time for \emph{any} resource to find a PoW for $(d,c)$ is some constant~$\tau$.

\paragraph*{The Blockchain Protocol}
\label{sec:blockchain-protocol}

The blockchain is a collaboratively maintained list whose function is to throttle joins of new identities to the CA protocol by employing a PoW mechanism.
A single \emph{block} in the blockchain has the form
\newcommand{\Block}[1][]{\ensuremath{\langle h_{#1}, d_{#1}, p_{#1}, x_{#1} \rangle}}
\[
b = \Block,
\]
where $h$ is a hash value, $d$ is a difficulty, $p \in P$ is a peer, and $x$ is a bit-string.
We denote by $\mathcal H$ the hash function used to calculate $h$.
A \emph{blockchain} consists of a sequence $C = (b_1, \dots, b_l)$ of blocks, and a \emph{genesis block} $b_0$ that is fixed in advance.
From here on, we assume the system implementation provides an agreed-upon genesis block.

For $i \geq 1$, block $b_i = \Block$ is said to be \emph{legal} if
\[
\begin{array}{rcl}
  h                                      & = & \mathcal H(b_{i-1}), \text{ and } \\
  \mathcal F(d, \langle h, p \rangle, x) & = & \text{true},
\end{array}
\]
that is, if the hash in $b_i$ is obtained from $b_{i-1}$, and $b_i$ is a Proof-of-Work.
For a legal block $b_i$, the block $b_{i-1}$ is called the \emph{parent} of $b_i$, and $b_i$ is a \emph{child} of $b_{i-1}$.
A blockchain is \emph{legal} if every non-genesis block is legal.

\SetKwFunction{mine}{mine}
\SetKwFunction{proposeBlock}{propose\_block}
Since the blockhain is based on a PoW mechanism it is ensured that new blocks cannot be appended to $C$ at will.
Attempting to find a legal block that extends the current blockchain is called \emph{mining}.
We encapsulate this process in the procedure $\mine(b)$, which for peer $p$ attempts to find a block with parent $b$ that includes $p$'s identity.

%

Note that legal blocks together with $b_0$ form a tree rooted at $b_0$ due to the parent/child relation, and a legal blockchain corresponds to a path in the tree starting at the root.
In order to provide forward security, it is necessary that once the peers agree on a blockchain $C$, they will never accept a blockchain that does not have $C$ as a prefix.
To tackle this issue, whenever the blockchain is extended the CA protocol is used to ensure that all peers agree on the same extended blockchain.
In particular, the BC protocol relies on the \proposeBlock operation provided by the Chain Agreement.

If the Chain Agreement protocol accepts the block proposed by peer $p$, then the identity of $p$ becomes voting.
In that case the resources allocated to $p$'s mining process by the controlling entity of $p$ may be assigned to a new identity.
If on the other hand a block containing a different peer is accepted, then $p$ continues mining and proposes the next block it finds.
Refer to \cref{fig:bc-algo} for a pseudo-code description of the BC protocol.

\begin{figure}[h]
  \begin{algo}
    \Protocol{Blockchain, from the perspective of peer $p$}{
      \Init{
        $C \gets $ the current Blockchain, obtained from CA\;
        trigger Start event \;
      }
      \OnEvent{Start}{
        $b \gets $ the newest block in $C$ \;
        $\mine(b)$ \;
      }

      \OnEvent{\mine($b$) returns block $b^\ast$}{
        $\proposeBlock(b^\ast)$ using CA \;
      }

      \OnEvent{CA commits a block $a$}{
        stop mining \;
        $C \gets$ the new blockchain from CA \;
        \If{$a \not= b^\ast$}{
          trigger Start event \;
        }
      }
    }
  \end{algo}
  \caption{The Blockchain Protocol.}
  \label{fig:bc-algo}
\end{figure}

\subsection{Chain Agreement (CA)}
\label{sec-3-3}




While the blockchain introduces new identities into the system, the Chain Agreement tracks the membership of currently participating identities in the system.
For our CA protocol we adapt SGMP~\cite{reiter1996securegroup} and the PBFT~\cite{castro1999practical} agreement protocols.
In particular, the goal is to keep track of some \emph{shared state} that can be modified by certain predetermined \emph{operations}.
In our case, the shared state encompasses an \emph{operation log} $O$, a set of \emph{online voters} $I$, and the blockchain $C$.

As in SGMP and PBFT, the life cycle of an operation $op$ begins with $op$'s \emph{proposal}.
The proposal is sent to the \emph{primary}, i.e., to a specific peer determined by an agreed-upon scheme.
Given that $op$ is valid and the peers decide to commit it, $op$ is applied to the shared state.
Both agreement protocols rely on the notion of totally ordered \emph{logical time stamps}, and in each such time step exactly one operation is committed.
A \emph{logical time stamp} is a triple $(\ell, v, s)$, where $\ell$ is the current length of $C$ (i.e., the blockchain contained in the shared state), and $v$ and $s$ are positive integers referred to as the \emph{view primary number} and \emph{sequence number}, respectively.
Logical time stamps are ordered in lexicographic order.

To determine the primary we introduce the notion of a peer's rank.
For a fixed blockchain $C = (b_1, \dots, b_\ell)$ and a voting peer $p$ let $i$ denote the index of the block in which $p$ appears.
The \emph{rank of $p$}, denoted by $\rank(C, p)$, is $\ell - i$, i.e., peers are ranked by how recently the right to vote was obtained.
Note that the rank is well defined since a peer can acquire the right to vote only once.

Consider a time stamp $(\ell, v, s)$ and the associated blockchain $C$ of length $\ell$.
The peer $p$ with $\rank(C,p) = v \pmod \ell$ is chosen as the primary, i.e., the peer who accepts operation proposals for the next time step.
We use the failover mechanism of PBFT to ensure that $v$ is increased without the help of a primary in case the current primary fails.

Using the logical time stamps and the rank as fixed above, the underlying SGMP/PBFT agreement protocols can be used to implement Chain Agreement.
Note however that due to churn, just like SGMP, CA cannot support a snapshot mechanism.
This is in contrast to PBFT where the set of participating peers is fixed in advance and snapshots are supported.

\paragraph*{Operations}

The Chain Agreement uses a standard byzantine agreement technique, in which each operation has to go through the stages propose, pre-prepare, prepare, and commit before it is applied.
More specifically, operations are initially proposed to the current primary $q$.
The task of $q$ is to assign consecutive time stamps to proposed operations.
For each proposal, $q$ then sends out pre-prepare messages, receives prepare messages, and commits the operation once $q$ received a sufficient amount of prepare messages from peers in $I$.
Recall that in each step, authenticity of messages is guaranteed due to signatures offered by the public key cryptography system.

What is left in the Chain Agreement specification are the operations mutating the shared state.
The Chain Agreement protocol relies on the following three operations:

\SetKwFunction{block}{block}
\SetKwFunction{join}{join}
\SetKwFunction{leave}{leave}

\begin{itemize}
 \item $\block(b)$
   is used to append a new block $b$ to the Blockchain, thus promoting the peer contained in $b$ to be promoted to voting.
 \item $\join(p)$
   is used by a previously offline voting peer $p$ to re-join the set $I$ of online voters.
 \item $\leave(p)$
   is used to remove offline peers from $I$.
\end{itemize}

We need to explicate two aspects of each operation, namely how the operation \emph{validated}, and how \emph{committing} it affects the shared state.
Validation occurs at the primary when an operation is proposed, and at other nodes upon receiving a pre-prepare message for that operation.
This is to ensure that a faulty/malicious user cannot modify the shared state in an undesired manner.
Whenever an operation is committed, peers append the operation together with its assigned time stamp and collected commit signatures to the operation log and update their new time stamp accordingly.
Furthermore, committing an operation may modify the shared state according to the operation's purpose.
We now describe both aspects for each operation separately and refer to \cref{fig:ca-protocol} for a pseudo-code description.

\begin{figure}[t]
  \input{ca-algo.tex}
  \caption{Operations of the Chain Agreement Protocol}
\label{fig:ca-protocol}
\end{figure}


Recall that proposals for a block $b$ are sent to the Chain Agreement only from the Blockchain layer.
To validate a $\block(b)$ operation, all peers check that $b$ is indeed valid and extends the current blockchain $C$.
To commit this operation $b$ is appended to $C$, and the time stamp is set to $(\ell, 0, 0)$, where $\ell$ is the new blockchain length.
This results in the block finder becoming the new primary, with the previous primary as backup.


A join operation consists of the joining peer $p$.
To validate a join, peers check whether $p$ is indeed reachable over the network.
In that case, the operation will be committed and $p$ is included in the set $I$.


Peers rely on a failure detector to detect when identities left the system, e.g., by sending \emph{ping} messages in regular intervals.
Should one peer detect a failure of another peer $p$, a leave operation on behalf of $p$ will be emitted.
A $\leave(p)$ operation is validated by checking whether $p$ indeed failed, to keep malicious peers from removing online peers.
When the operation turns out to be valid, it is committed by removing $p$ from $I$.

\subsection{Application}
\begin{figure}
 \includegraphics{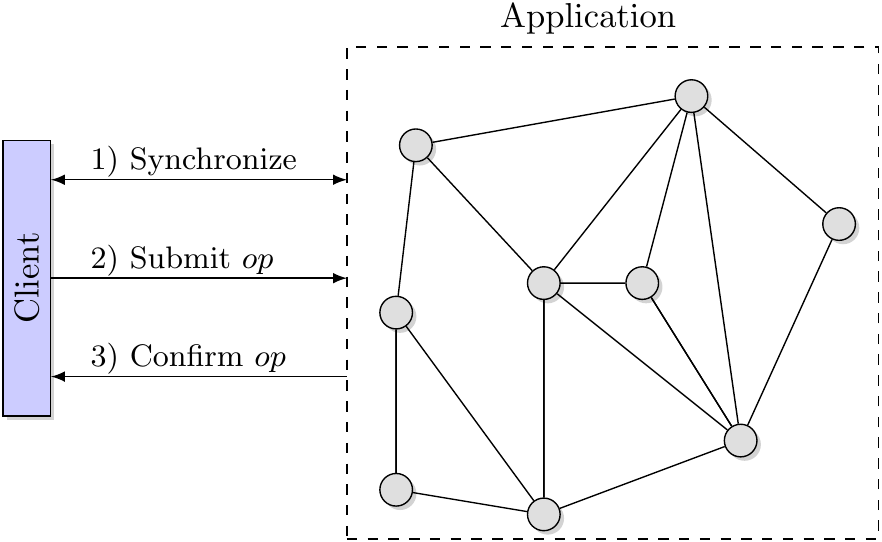}
 \caption{Application client communicating with the application.}
 \label{fig:client-service}
\end{figure}

The application layer makes use of the membership information from the CA in order to implement the application logic.
The CA provides a ranking among identities, the current membership as well as its timestamp, which enables the application to use the full capabilities of PBFT.
This includes the use of snapshots of the application state.

The application has some shared state and deterministic operations that modify the state.
Operations are totally ordered by assigning a timestamp $(t,o)$ to them, where $t$ is the membership timestamp from the CA and $o$ is an \emph{operation sequence number} assigned by the current primary.

The application logic and state is encapsulated in the application layer and does not influence the decisions in the CA.
A single instance of the CA and the BC can therefore be shared among any number of applications.

Applications may export functionality to clients that are not participating in the application agreement, like depicted in \Cref{fig:client-service}.
Clients synchronize with the CA in order to get the membership information.
The synchronization consists of downloading the CA operation and incrementally applying it to the membership.
The clients then submit operations to the application, which in turn processes them.
Using the membership information, the clients then verify the confirmation that the operation was processed correctly.

\hide{
\subsection{Protocol Discussion}
Compared to SGMP, the \PeerCensus protocol features a dynamic set of voting identities.
The change of the set of voting identities 
The primary change upon a block being committed is due to the ranking and set of voting identities changing.

Since only blocks in the blockchain are promoted, peers will only extend the longest chain.
By committing successor blocks in the CA the chain is guaranteed not to change in the future:
Even if an attacker eventually controls a majority of identities, it will not be able to forge the votes to replace a successor block, hence the blockchain is forward secure.


For simplicity entities create a new identity for each block they attempt to find.
As an optimization it is also possible to assign weights to the identities, which are incremented each time the corresponding identity is included in a block.
The protocol would then switch to a weighted voting.
}

\section{Safety \& Liveness}
\label{sec:analysis}

We would like to lift the safety and liveness guarantees provided by PBFT~\cite{castro1999correctness} and apply them to our Chain Agreement.
An agreement protocol provides \emph{safety} if operations on the shared state are committed atomically, i.e., as if they were applied on a single sequential machine;
An agreement protocol provides \emph{liveness} if all proposed valid operations are eventually committed.
The premise under which PBFT provides both is that less than one third of the participants are not faulty.


In our setting participants in the \emph{protocol} are modeled as peers, whereas participants in the \emph{system}, i.e., a individuals with an agenda to subvert the protocol, are modeled as \emph{entities}.
In order to lift the guarantees from PBFT to Chain Agreement, we need to ensure that at any time $t$, less than one third of the online voters (the set $I$ in the CA) are controlled by a single entity.
Since SGMP ensures that $I$ tracks the voters in $P(t)$ (with some delay depending on the message delays and failure detector speeds, cf.~\cite{reiter1996securegroup}), it is sufficient to investigate how $P(t)$, and in particular the voters therein, evolves over~time.

To state this formally, let $A$ be a malicious entity referred to as \emph{attacker}.
To simplify the analysis, we denote by $D$ a meta-entity that encompasses all entities that are not $A$.
For some fixed point in time, let $I$ be the set of online voters.
We denote by $I_A$, and $I_D$ the corresponding partition of $I$ into online peers controlled $A$, and $D$, respectively.
We can apply the classic positive results for byzantine agreement due to Lamport~\cite{lamport1980reaching} if it holds that $|I_A|/|I| < 1/3$.
This is equivalent to ensuring that
\[
\phi_I := \frac{|I_A|}{|I_D|} < 1/2\,.
\]

Therefore, as long as the inequality remains satisfied we say that \PeerCensus is in a \emph{secure state}.
On the other hand, Lamport's work also established that no guarantees can be made should the inequality be exceeded.
Correspondingly, when the inequality is violated we say that \PeerCensus is in an \emph{insecure state}.

What are the consequences of being in an insecure state?
First observe that $A$ can cement its control by not committing block or join operations, thus hindering peers controlled by other entities from being included the online voter set.
The effect for the application layer is that new operations are only applied at $A$'s will.
Note however, that past committed operations cannot be modified or undone by any attack on the protocol, i.e., strong consistency up to the time when $A$ took control is still guaranteed.

Our analysis relies on the system being in its \emph{steady state}, i.e., that the number of online peers and resources is governed by the respective expected value.
This is the case if \PeerCensus was active for a sufficiently long time.
Later in \cref{sec:reach-stable-state} we show that this assumption is justified due to a bootstrapping method.
Before describing the procedure in detail, we now turn to establishing our following main theorem.

\begin{theorem}
  \label{theorem:main-theorem}
  Let $\phi_R$ denote the fraction of resources associated with $A$ over resources not associated with $A$, and let $0 < \epsilon < 1/2$ be a constant.
  If \PeerCensus reaches a steady state and $\phi_R < 1/2 - \epsilon$, then \PeerCensus is in a secure state with high probability.
\end{theorem}

To prove \cref{theorem:main-theorem} we separately consider the three factors that influence the cardinalities of $I_A$ and $I_D$, namely membership churn, resource churn and miner's luck.

\begin{itemize}
\item
  \emph{Resource churn}:
  Resources fail and recover, thus limiting or enhancing the attacker's capability to introduce new peers to the voter set.
\item
  \emph{Membership churn}:
  Voting peers fail and recover, directly affecting $I_A$ as well as $I_D$.
\item
  \emph{Miner's luck}:
  A stochastic block mining process determines who gets to introduce a new peer to the voter set.
  With non-zero probability, an attacker's resources may mine more blocks than expected, thus increasing $P_A$ disproportionately.
\end{itemize}

\subsection{Preliminaries}
\label{sec:bound-churn}

In the steady state, resource churn is characterized by a parameter $\rho$ in the following way.
The state of an individual resource is modeled as a two-state Markov-Chain with the transition matrix
\[
\left(
  \begin{array}{cc}
    1-p & p \\
    q & 1-q
  \end{array}
\right)\,,
\]
where $p$ and $q$ denote the probability of a resource to fail or recover, respectively.
The two states indicate whether the resource is currently active, or inactive.
For a single resource, the stationary distribution is $(\rho, 1-\rho)$, where $\rho = q/(p+q)$.
We conclude that in the steady state the expected number of online resources is $\rho |R|$, since resources fail or recover independently from one another.

\begin{lemma}
  \label{lemma:resource-probability}
  Let $\phi_R$ be the random variable representing the ratio of online resources for $A$ to online resources for $D$.
  In the steady state and for $\alpha \in (0,1/2)$ it holds that
  \begin{alignat*}{2}
    \Pr\left[\phi_R \geq \left(1+\frac{2\alpha}{1-\alpha}\right) r\right]
    & <
    &
    & \left(\frac{\exp(\alpha)}{(1+\alpha)^{1+\alpha}}\right)^{\rho nr/(1+r)}
    \\
    &
    & +
    &
    \left(\frac{\exp(-\alpha)}{(1-\alpha)^{1-\alpha}}\right)^{\rho n / (1+r)}\,,
  \end{alignat*}
  where $n$ is the cardinality of $R$, and $r$ is the ratio of $A$'s resources to $D$'s resources in $R$.
\end{lemma}

\begin{proof}
  Denote by $R_A \dot\cup R_D = R$ the partition of $R$ into resources belonging to the attacker $A$ and defender $D$.
  For $i \in R_A$, let $X_i$ be the 0/1 random variable indicating whether resource $i$ is online.
  Correspondingly for $j \in R_D$, let $Y_j$ be the 0/1 random variable indicating whether resource $j$ is online.
  Let $X$ and $Y$ be the corresponding random variables denoting the sum of $X_i$ and $Y_j$.
  Note that in the stationary distribution, the expected value of $X$ and $Y$ are $\rho |R_A|$ and $\rho |R_D|$, respectively.

  With these definitions $\phi_R = X/Y$.
  Since $X$ and $Y$ are independent it holds that $E[\phi_R] = E[X]/E[Y] = r$.
  Our goal is to bound the probability that $\phi_R$ deviates from its expected value by bounding the probability of $X$ and $Y$ deviating from their expected values.
  Applying the Chernoff bound (see, e.g., \cite{mitzenmacher2005probability}) to $X$ and $Y$ yields that
  \begin{align*}
    \Pr[X > (1+\beta)\rho |R_A|]  &
    < \left(\frac{\exp(\beta)}{(1+\beta)^{1+\beta}}\right)^{\rho |R_A|}\,\text{, and} \\
    \Pr[Y < (1-\gamma)\rho |R_D|] &
    < \left(\frac{\exp(-\gamma)}{(1-\gamma)^{1-\gamma}}\right)^{\rho |R_D|}\,
  \end{align*}
  for any $\beta > 0$ and $0 < \gamma < 1$.
  Let $\event X(\beta)$ and $\event Y(\gamma)$ denote the two events from above, i.e., that $X$ resp.\ $Y$ deviates from the corresponding expected value by $(1+\beta)$ and $(1-\gamma)$.

  Let $\event Z$ denote the event that $\phi_R > (1+2\alpha)r$, i.e., the event from the statement, and consider positive values $\beta$ and $\gamma$ such that $\beta + \gamma = 2\alpha$.
  If neither $\event X(\beta)$ nor $\event Y(\gamma)$ occurs, then also $\event Z$ does not occur.
  By applying the union bound we obtain
  \[
  \Pr[\event Z]
  \leq \Pr[\event X(\beta) \lor \event Y(\gamma)]
  \leq \Pr[\event X(\beta)] + \Pr[\event Y(\gamma)]\,.
  \]
  We bound the above by applying the previously obtained Chernoff bounds for $\event X(\beta)$ and $\event Y(\gamma)$.
  Doing so yields
  \[
  \Pr[\event Z]
  < \left(\frac{\exp(\beta)}{(1+\beta)^{1+\beta}}\right)^{\rho |R_A|}
  + \left(\frac{\exp(-\gamma)}{(1-\gamma)^{1-\gamma}}\right)^{\rho |R_D|}\,.
  \]
  This resulting sum is minimized if $\beta = \gamma$, i.e., $\alpha = 2\beta/(1-\beta)$.
  By observing that $|R_A| = nr/(1+r)$ and $|R_D| = n/(1+r)$ the proof is completed.
\end{proof}

\Cref{lemma:resource-probability} bounds the impact of resource churn.
Our next goal is to do the same for membership churn.
To that end, similar to the discussion above, we characterize the membership churn in the steady state by the constant $\sigma = p_{pr}/(p_{pr} + p_{pf})$.

\begin{lemma}
  \label{lemma:membership-probability}
  Let $\phi_I$ be the random variable representing the ratio of online voters for $A$ to online voters for $D$.
  In the steady state and for $\alpha \in (0,1/2)$ it holds that
  \begin{alignat*}{2}
    \Pr\left[\phi_I \geq \left(1+\frac{2\alpha}{1-\alpha}\right) s\right]
    & <
    &
    & \left(\frac{\exp(\alpha)}{(1+\alpha)^{1+\alpha}}\right)^{\sigma ns/(1+s)}
    \\
    &
    & +
    &
    \left(\frac{\exp(-\alpha)}{(1-\alpha)^{1-\alpha}}\right)^{\sigma n / (1+s)}\,,
  \end{alignat*}
  where $n$ is the cardinality of $I$, and $s$ is the ratio of $A$'s peers to $D$'s peers in $P$.
\end{lemma}

The above lemma can be established using the same techniques as in the proof of \cref{lemma:resource-probability}.
We therefore omit the proof here.
Note that the parameter $s$ in \cref{lemma:membership-probability} is directly affected by the outcome of the block mining process.
Before establishing our main theorem we thus derive bounds on the miner's luck of the attacker in the following lemma.

\begin{lemma}
  \label{lemma:block-probability}
  Let $\phi_B$ be the random variable representing the ratio of $A$'s blocks to $D$'s blocks in the blockchain.
  In the steady state and for $\alpha > 0$ it holds that
  \[
  \Pr[\phi_B \geq (1+\alpha) t]
  \leq \left(\frac{\exp(\alpha)}{(1+\alpha)^{1+\alpha}}\right)^{\ell t}
  \]
  where $\ell$ is the current length of the blockchain, and $t$ is the fraction of $A$'s resources in $R$.
\end{lemma}

\begin{proof}
  Let $X_i$ be the 0/1 random variable indicating whether the attacker found block $i$, and let $X$ denote its sum.
  It holds that $E[X] = \ell t$, since the resource that found block $i$ is drawn uniformly at random from the online resources, and in the steady state a $t$-fraction of those belongs to $A$.
  By the Chernoff bound,
  \[
  \Pr[X \geq (1+\alpha)\ell t]
  \leq \left(\frac{\exp(\alpha)}{(1+\alpha)^{1+\alpha}}\right)^{\ell t}\,.
  \]

  Since $\ell \phi_B \geq X$, the probability of the event $\ell \phi_B \geq (1+\alpha)\ell t$ is upper bounded by the same term.
  Dividing by $\ell$ concludes the proof.
\end{proof}

Note that the expected value of $\phi_B$ is not $t$---it rather depends on the resource distribution between $A$ and $D$.
Suppose that $E[\phi_B] = u$, and set $\alpha = (u\alpha' - t + u)/t$ for some $\alpha' > 0$.
Since $\alpha' > 0$ implies $\alpha > 0$, we may apply \cref{lemma:block-probability} to obtain the following technical corollary, which is the last building block for our proof of \cref{theorem:main-theorem}.

\begin{corollary}
  \label{corr:block-probability}
  Let $\phi_B$ be the random variable representing the ratio of $A$'s blocks to $D$'s blocks in the blockchain.
  In the steady state and for $\alpha' > 0$ it holds that
  \[
  \Pr[\phi_B \geq (1+\alpha') E[\phi_B]]
  \leq \left(\frac{\exp(\alpha)}{(1+\alpha)^{1+\alpha}}\right)^{\ell t}
  \]
  where $\ell$ is the current length of the blockchain, $t$ is the fraction of $A$'s resources in $R$, and $\alpha = (E[\phi_B]\alpha' - t + E[\phi_B])/t$.
\end{corollary}

\subsection{Establishing \cref{theorem:main-theorem}}
\label{sec:safety-liveness-conclusion}

Let $\epsilon < 1/2$ be a positive constant.
The goal is to show that if $\phi_R < 1/2 - \epsilon$, then with high probability%
 the Chain Agreement is in a secure state.
To that end, consider the complementary event $\event T$ that the CA reaches an insecure state.
We establish the claim by showing that $\event T$ occurs with probability at most $\exp(-\Omega(\min(|R|,|I|,\ell)))$, where $R, I$, and $\ell$ are as above.

Let $\alpha, \beta, \gamma$ be positive constants with $\alpha + \beta + \gamma = \epsilon$.
We would like to use \cref{lemma:resource-probability}, \ref{lemma:membership-probability}, and \cref{corr:block-probability} to obtain the result.
To apply those three we perform a worst case analysis:
Consider the event $\event{U}$ that after reaching the steady state, $\phi_R, \phi_B$, or $\phi_I$ deviate from their expected value by more than $\alpha, \beta$, or $\gamma$, respectively.
Note that $\event U$ occurring is necessary, but not sufficient, for $\event T$ to occur.

Event $\event U$ corresponds to the occurrence of at least one of the events bounded in \cref{lemma:resource-probability}, \ref{lemma:membership-probability}, and \cref{corr:block-probability}.
Thus, applying the union bound to $\event U$ yields
\begin{alignat*}{2}
  \Pr[\event T] \leq \Pr[\event U]
  & \leq
  &
  & \Pr\left[\phi_R \geq \left(1+\alpha\right) E[\phi_R]\right] \\
  & & + & \Pr\left[\phi_B \geq \left(1+\beta\right) E[\phi_B]\right] \\
  & & + & \Pr\left[\phi_I \geq \left(1+\gamma\right) E[\phi_I]\right]\,.
\end{alignat*}

The statements of the two lemmas and the corollary can now be used to bound the three corresponding terms.
This concludes our proof of \cref{theorem:main-theorem}.\hfill\IEEEQED


\subsection{Reaching the Steady State}
\label{sec:reach-stable-state}
The security of the system hinges on it starting in a steady state, i.e., that there are a sufficient number of resources, voting identities and online peers.
For example should no identity have been promoted yet, then the first block finder controls all identities in the system, trivially subverting the system.
A bootstrapping period is used to ensure a large enough initial number of resources and voting identities set, resulting in good bounds on the failure probability.
In order to reach a steady state it is necessary to bootstrap the system in a controlled way.
Bootstrapping consists of determining a genesis block, an initial set of voting identities and an initial set of online~identities.

\PeerCensus can be bootstrapped by retrofitting the Bitcoin blockchain, providing the initial resources, blocks (voting identities) and peers.
Every block in Bitcoin contains a \emph{reward-transaction}, transferring a fixed amount of newly minted Bitcoins to the block finder.
In order to receive the Bitcoins, the block finder has to include a Bitcoin addresses in the transaction.
This enables us to derive the new voting identity from the block by extracting the Bitcoin address from the reward transaction.

\begin{figure}[t]
  \centering
  \includegraphics[width=0.5\textwidth, trim = 8mm 0mm 12mm 6mm, clip]{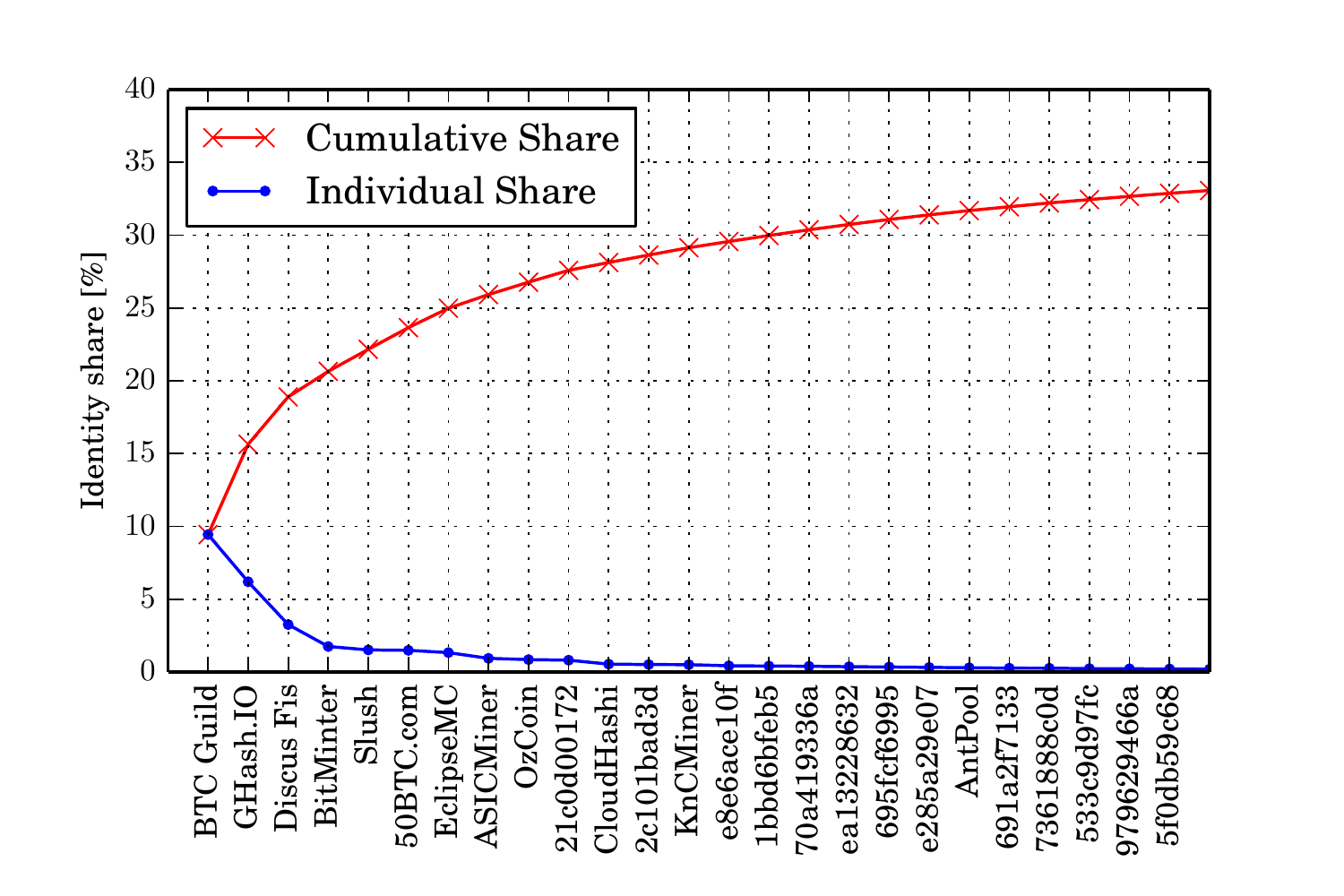}
  \caption{Bitcoin block finder distribution as of blockchain height 333,000 for the 25 most prominent mining pools.}
  \label{fig:block-finders}
\end{figure}

To migrate from Bitcoin to \PeerCensus a migration time in the form of a blockchain length $l_m$ is determined in advance.
Garay et al.~\cite{garay2014backbone} showed that with high probability peers agree on a common prefix, with distance $k$ from the current blockchain head and that the blockchain of length $j$ is a representative sample of online peers with high probability.
Upon receiving a valid block for blockchain length $l_m$, peers extract the identities from blocks $[0, l_m-k]$.
The Bitcoin genesis block is also the \PeerCensus genesis block.
The initial set of online identities is then assumed to consist of the last $j$ identities, i.e., the identities included in blocks $[l_m-k-j, l_m-k]$.
The parameter $j$ should be chosen small enough so that $\ceil*{2j/3}+1$ identities are online to guarantee liveness, but large enough to ensure diversity in the entities.
Once the set of voting and online voting identities are determined, peers start executing the \PeerCensus protocol.
The peers then incrementally commit blocks at heights $[l_m-k, l_m]$.

The migration requires that in Bitcoin's current blockchain there is no entity that has mined a sufficient number of blocks to take control of the system.
Fortunately, many mining pools include identifying hints in blocks, e.g., reusing the same address or including a text banner, so that the blocks can be attributed to the pool.
This allows us to determine the block finder of a large percentage of blocks found so far in the blockchain.
\Cref{fig:block-finders} shows the current shares of blocks found by mining pools and therefore their share of identities in \PeerCensus.
Even if the largest 28 pools were to collaborate they would not reach a sufficient share of blocks to take control of the system.
Furthermore, with $j\geq 10,000$ there is no single entity that controls more than $25\%$ of identities, securing the migration itself.

\subsection{Real World Guarantees}
\label{sec:real-world-guar}

The previous subsections established that with high probability the system does not fail, for increasing number of resources and identities.
In this section we give an example of the guarantees that are to be expected in real world instances of the \PeerCensus system.
In order to gauge the probability of a failure of the system we need to estimate some parameters used in the analysis.

For the resources we need to determine a maximum ratio of resources an attacker is allowed to control $25\%$ which results in a security margin of $\epsilon=1/2-1/3$.
Notice that this is equivalent to the 13 largest mining pools colluding to subvert the system according to \Cref{fig:block-finders}.
The number of resources is estimated as $1,000,000$, which at the current computational power in the Bitcoin network of $274,000,000 GH/s$ (Gigahashes) would mean that a unit resource has $274 GH/s$, which matches the currently available ASIC mining hardware.
The number of blocks in the system is estimated as $350,000$ blocks, matching the Bitcoin blockchain length.
The number of peers that are online in expectation is estimated at $25,000$ peers.
Furthermore we adopt a conservative mean time between failures of 99 days and a mean time to recovery of 1 day for resources and peers, resulting in $\rho=\sigma=0.99$.
Applying \cref{theorem:main-theorem} using these parameters yields the following upper bound on the failure probability of
\[
\Pr[\text{\PeerCensus is in a secure state}] \geq 1-4.26\cdot 10^{-15}
\]
in one time interval. Notice that this results from subdividing the security margin $\epsilon$ as $2\alpha_R=14\%\epsilon$, $\alpha_M=11\%\epsilon$ and $2\alpha_I=75\%\epsilon$.
If the system proceeds in discrete time intervals of 1 second, then the system therefore is expected to fail fewer than once every 7 million years.

\section{Discoin}
In the following we present Discoin, a crypto\-currency, as an exemplary application built on \PeerCensus.
Discoin tracks the balances of \emph{accounts}, denominated in \emph{coins}.
An account $a$ is associated with a public-/private-keypair $(p_a,s_a)$.
The public key $p_a$ is used to identify the account, while the private key $s_a$ is used to authenticate messages.
%

\SetKwFunction{process}{transaction}

The shared state in Discoin consists of account balances~$\mathcal{B}$.
In order to transfer coins between accounts we define a \emph{transaction} $tx=\langle a,b,v\rangle_\sigma$.
A transactions describes a transfer of $v$ coins from source account $a$ to destination account $b$ and includes signature $\sigma$ by the private key of $a$ to authorize the transfer.
A transaction is \emph{valid} if the source account's balance $\mathcal{B}[a] \geq v$, the signature $\sigma$ correctly signs $\langle a,b,v \rangle$ and matches the public key of $a$.

Discoin has a single operation $\process(tx)$ which, if committed, applies the transaction to the account balances.
Upon committing a $\process(\langle a,b,v\rangle_\sigma)$ operation the value is subtracted from the source account's balance and added to the destination account.
\begin{figure}[t]
  \begin{algo}
    \Specification{Discoin Transaction processing}{
      \State{
        $\mathcal{B}$ \Comment*{\makebox[5.5cm]{Account balances\hfill}}
      }

      \Validate{$\process(\langle a,b,v\rangle_\sigma)$}{
        \eIf{$\sigma$ is valid signature by $s_a$ and $\mathcal{B}[a] \geq v$}{\Return valid}{\Return invalid}
      }

      \OnCommit{$\process(\langle a,b,v\rangle_\sigma)$}{
        $\mathcal{B}[a] \leftarrow \mathcal{B}[a] - v$ \;
        $\mathcal{B}[b] \leftarrow \mathcal{B}[b] + v$ \;
      }
    }
  \end{algo}
  \caption{Discoin protocol}
\label{fig:discoin-protocol}
\end{figure}
Finally, Discoin distributes a reward of $r$ newly generated coins each time a block is found.
The $r$ coins are distributed in equal parts to each identity $i \in I$.
This reward is triggered by the timestamp change and does not necessitate a new transaction.
By using PBFT we are guaranteed to process the transactions in the same order.
The peers agree on the validity of individual transactions and on the current balance of each account.

Compared to Bitcoin, Discoin features a much leaner and simpler protocol.
Unlike Bitcoin which tracks transaction outputs, we explicitly track account balances which results in a smaller shared state and a more intuitive concept of account balances.
Committing a transaction is independent from the block generation and, more importantly, once transactions are committed they stay committed.
By distributing rewards among all participants instead of just the block finder, Discoin continuously incentivizes peers to participate in the network.
This contrasts Bitcoin's all-or-nothing rewards, which incentivize the creation of mining pools which pool resources and distribute the reward.
Mining pools are seen as single points of failure in the Bitcoin ecosystem~\cite{eyal2013majority,miller2014nonoutsourceable,rosenfeld2011analysis,rosenfeld2012analysis}.

As with the bootstrapping of \PeerCensus, the accounts from Bitcoin can be migrated to Discoin.
Once \PeerCensus is bootstrapped, Discoin can be bootstrapped by computing the account address balances up to Bitcoin's blockchain height $l_m$.
A snapshot of the balances is then committed before proceeding with the Discoin protocol and committing new transactions.

\section{Related Work}
\label{sec:related-work}

The study of byzantine agreement protocols was initiated by the seminal works by Lamport et al.~\cite{lamport1982byzantine,lamport1980reaching}, establishing tight feasibility results.
\PeerCensus and Discoin rely on byzantine agreement protocols that later improved message complexity, e.g., PBFT~\cite{castro1999practical}, Zyzzyva~\cite{kotla2007zyzzyva} and SGMP~\cite{reiter1996securegroup}.

Bitcoin~\cite{nakamoto2008bitcoin} is the latest, and most successful, in a long series of attempts to create a decentralized digital currency initiated by DigiCash~\cite{chaum1983blind} and ECash~\cite{chaum1990untraceable} by David Chaum.
Recent work by Garay et al.~\cite{garay2014backbone} and Miller et al.~\cite{miller2014anonymous} has shown that, with high probability, the peers participating in the Bitcoin network eventually agree on a transaction history.
Reaching consistency however is a slow process as blocks are counted as individual votes for the validity of a transaction and confirmations are never final.
%
%
Committing blocks in the CA resolves blockchain forks~\cite{decker2013information} early, rather than deferring the resolution to a later time, shown by Garay et al.~\cite{garay2014backbone} to be inefficient.

Today, a multitude of altcoins, i.e., alternative cryptocurrencies~\cite{miers2013zerocoin,rosenfeld2012colored} and so called Bitcoin 2.0 projects~\cite{butterin2014next,clark2012commitcoin,miller2014permacoin,willet2012mastercoin}, are being used, each one using their own blockchain.
This splits available resources and mining efforts, weakening the individual blockchains.
Back et al.~\cite{back2014sidechains} proposed two-way pegged sidechains as a way to allow altcoins to be pegged to Bitcoin and to trade among altcoins, however each altcoin still has the burden of securing their own system via a blockchain.

Rosenfeld~\cite{rosenfeld2011analysis} analyses the difficulty of fairly distributing rewards among mining pool participants.
Pools have become powerful entities often acting selfishly~\cite{babaioff2012bitcoin}.
Eyal and Sirer~\cite{eyal2013majority} show that a mining pool may increase its share by not publishing blocks immediately.
Miller~\cite{miller2014nonoutsourceable} propose an alternative proof of work mechanism that would not allow pools to form.

Schwartz et al.~\cite{schwartz2014rippleconsens} describe how consensus in Ripple is achieved by unique node lists assumed not to collude.
Maintaining the node lists however requires manual configuration in order to avoid sybil attacks.

\PeerCensus solves problems arising from inconsistent state views, such as double-spendings~\cite{bamert2013snack,karame2012two}. It does not address problems like transaction malleability~\cite{decker2014malleability} and privacy issues, e.g.,~\cite{androulaki2012evaluating,reid2011anonymity}.


\hide{
\section{Conclusion}

Retrofitting Bitcoin to use \PeerCensus not only shortens the bootstrapping period by reusing blocks and existing infrastructure, but it also allows to maintain the momentum and the public acceptance Bitcoin has gathered over the years.
\PeerCensus is not yet another altcoin, it is a subsystem that unifies cryptocurrencies.
}
\bibliographystyle{plain}
\bibliography{references.bib}

\end{document}

%% file: ca-algo.tex
  \begin{algo}
    \Specification{Operations for Chain Agreement}{
      \State{
        \newlength{\cmtWidth}
        \setlength{\cmtWidth}{5.5cm}
        $O$ \Comment*{\makebox[\cmtWidth]{The operation log\hfill}}
        $I$ \Comment*{\makebox[\cmtWidth]{The set of online voters\hfill}}
        $C$ \Comment*{\makebox[\cmtWidth]{The blockchain\hfill}}
        $t = (\ell, v, s)$ \Comment*{\makebox[\cmtWidth]{The logical time stamp\hfill}}
      }

      \Validate{$\block(b)$}{
        $b' \gets $ the newest block in $C$ \;
        \eIf{$b$ is a child of $b'$ and $b$ is legal}{\Return valid}{\Return invalid}
      }

      \OnCommit{$\block(b)$}{
        Append $\block(b)$ to $O$ \;
        Append $b$ to $C$ \;
        $\Block \gets b$ \;

        $\onlinevoting{} \gets \onlinevoting{} \cup \{p\} $ \hspace*{1.76cm}\Comment{Promote $p$ to voting}
        $\ell \gets $ the length of $C$\hspace*{0.79cm}\Comment{Update logical time stamp}
        $v \gets 0$ \;
        $s \gets 0$ \;
      }

      \Validate{$\join(p)$}{
        Send a ping message to $p$ \;
        $V \gets $ the set of peers appearing in the blocks of $C$ \;
        \eIf{$p \in V$, $p \not\in I$, and $p$ replies to the ping}{\Return valid}{\Return invalid}
      }

      \OnCommit{$\join(p)$}{
        Append $\join(p)$ to $O$ \;
        $\onlinevoting{} \gets \onlinevoting{} \cup \{p\} $\;
      }

      \Validate{$\leave(p)$}{
        Send a ping message to $p$ \;
        \eIf{$p \in I$ and $p$ does not reply}{\Return valid}{\Return invalid}
      }

      \OnCommit{$\leave(p)$}{
        Append $\leave(p)$ to $O$ \;
        $\onlinevoting{} \gets \onlinevoting{} \setminus \{p\} $\;
      }
    }
  \end{algo}
